\documentclass[letterpaper,journal,comsoc]{IEEEtran}

\usepackage[T1]{fontenc}
\usepackage{amsmath,amsfonts}
\usepackage[cmintegrals]{newtxmath}
\usepackage{bm}
\usepackage{algorithmic}
\usepackage{algorithm}
\usepackage{array}
\usepackage[caption=false,font=normalsize,labelfont=sf,textfont=sf]{subfig}
\usepackage{textcomp}
\usepackage{stfloats}
\usepackage{url}
\usepackage{verbatim}
\usepackage{graphicx}
\usepackage{cite}
\usepackage[colorlinks,anchorcolor=blue,linkcolor=blue,urlcolor=blue,citecolor=blue]{hyperref}
\usepackage{booktabs}
\usepackage{makecell}
\usepackage{multirow}
\usepackage{tabularx}
\usepackage{epsfig}
\usepackage{orcidlink}

\usepackage{amssymb}
\usepackage{amsthm}

\newcommand{\method}{QP-BD3QN-RACH}
\newcommand{\Nc}{N_{\mathrm{CBRA}}}
\newcommand{\Bset}{\mathcal{B}}
\newcommand{\Aset}{\mathcal{A}}

\newcommand{\Wset}{\mathcal{W}}
\newcommand{\Gset}{\mathcal{G}}
\newcommand{\Tset}{\mathcal{T}}
\newcommand{\Dset}{\mathcal{D}}
\newcommand{\qvec}{\mathbf{q}}

\newcommand{\xvec}{\mathbf{x}}
\newcommand{\svec}{\mathbf{s}}
\newcommand{\avec}{\mathbf{a}}
\newcommand{\hvec}{\mathbf{h}}

\newcommand{\E}{\mathbb{E}}

\newcolumntype{Y}{>{\raggedright\arraybackslash}X}

\newcommand{\etal}{~\textit{et~al}. }

\newtheorem{assumption}{Assumption}

\newtheorem{proposition}{Proposition}

\begin{document}

\title{QoS-Aware RACH Preamble Slicing via Quota-Projected Branching Deep Reinforcement Learning}

\author{
	Jiulin~Guo$^{\dagger\,\orcidlink{0009-0001-3908-8834}}$,~\IEEEmembership{Student Member,~IEEE}, Jiahan~Xu$^{\dagger\,\orcidlink{0009-0008-8987-1406}}$, Jiashuo~Zhang$^{\dagger\,\orcidlink{0009-0003-8596-4991}}$, Heng~Yang$^{\orcidlink{0000-0003-0417-6313}}$,~\IEEEmembership{Member,~IEEE}, Yizhen~Sun$^{\orcidlink{0009-0001-2477-2502}}$, Yutong~Xie$^{\orcidlink{0009-0008-9155-2383}}$, Shanshan~Li$^{\orcidlink{0000-0001-8563-7198}}$, Zhenyu~Liu$^{\orcidlink{0009-0001-0346-2433}}$, and Lei~Zhang$^{\orcidlink{0009-0008-4357-2144}}$%

	\thanks{This work was supported in part by the Natural Science Foundation of Liaoning Province under Grant 20180520022, Grant 2024MSLH359 and Grant 2025MSLH547; in part by the Scientific Research Fund of Liaoning Provincial Education Department under Grant LJ212410142021, Grant LJ222410142060 and Grant LJ212510142045; and in part by the University-Industry Collaborative Education Program of the Ministry of Education under Grant 231106665072727 and Grant 231106665030834. \textit{(Corresponding authors: Heng Yang; Shanshan Li.)}}%

	\thanks{Jiulin~Guo, Jiahan~Xu, Jiashuo~Zhang, Heng~Yang, Yizhen~Sun, Yutong~Xie, Zhenyu~Liu, and Lei~Zhang are with the School of Information Science and Engineering, Shenyang University of Technology, Shenyang 110870, Liaoning, China (e-mail: j.guo@smail.sut.edu.cn; jhxu@smail.sut.edu.cn; jszhang@smail.sut.edu.cn; h.yang@sut.edu.cn; yzsun@smail.sut.edu.cn; ytxie@smail.sut.edu.cn; liuzhenyu@sut.edu.cn; zhanglei@smail.sut.edu.cn).}%

	\thanks{Shanshan~Li is with the School of Mechanical Engineering, Shenyang University of Technology, Shenyang 110870, Liaoning, China (e-mail: sgdlss@sut.edu.cn).}%

	\thanks{$^{\dagger}$Jiulin Guo, Jiahan Xu, and Jiashuo Zhang contributed equally.}%
}

%\markboth{IEEE Transactions on Cognitive Communications and Networking,~Draft Manuscript}{Guo \MakeLowercase{\textit{et al.}}: QoS-Aware RACH Preamble Slicing via Quota-Projected Branching DRL}

\maketitle

\begin{abstract}
	Quality-of-service (QoS)-aware random access requires adaptive allocation of a finite random access channel (RACH) preamble budget across heterogeneous traffic and access procedures. This paper proposes QP-BD3QN-RACH, a quota-projected branching deep reinforcement learning controller for mixed two-step (2RA) and four-step (4RA) contention-based random access. Four action branches correspond to the delay-sensitive and delay-tolerant 2RA/4RA preamble pools. A branching dueling Double DQN selects pool-specific multipliers, and deterministic quota projection converts them to nonnegative integer allocations that preserve the preamble budget. With five actions per branch, the controller represents 625 pre-projection branch-action tuples using 20 branch-action outputs. Evaluation covers five arrival loads, cross-method comparison under nominal seed 42, six-seed sensitivity of QP-BD3QN-RACH, and targeted ablations. Across the five-load grid, its mean direction-aligned differences relative to four comparators are positive: 5.74 to 8.21 percentage points for success/collision, 1.23 to 1.92 percentage points for fallback, 0.35 to 0.68 percentage points for blocking, and 0.128 to 0.456 decision intervals for successful-access delay. Load-wise results exhibit metric-dependent tradeoffs, particularly under intermediate and overload conditions.
\end{abstract}


\section{Introduction}\label{sec:introduction}

\IEEEPARstart{M}{assive} Internet of Things (IoT) networks must accommodate dense sporadic access while sustaining timely network entry for latency-sensitive services. In New Radio (NR), contention-based random access (CBRA) allows unscheduled user equipment (UE) to initiate access by selecting random access channel (RACH) preambles, with contention resolution (CR) incorporated into the access procedure when multiple UEs select the same preamble. The 3GPP service-accessibility specification defines service-level requirements for heterogeneous services~\cite{3gpp.22.011}. The NR architecture specifies two-step random access (2RA) and four-step random access (4RA) procedures~\cite{3gpp.38.300}, while the NR medium access control (MAC) specification defines contention-resolution and retransmission behavior~\cite{3gpp.38.321}. When heterogeneous human-to-human (H2H) and machine-to-machine (M2M) arrivals share a finite preamble budget, repeated contention can propagate into fallback, backlog accumulation, and blocking. These characteristics motivate adaptive preamble allocation across traffic classes and access procedures.

Accordingly, a central line of RACH research has focused on adaptive resource control at the MAC layer. Turan\etal learn class-specific access class barring (ACB) rates for radio access network slices~\cite{turan2021adaptive_acb_slicing}; Lee\etal combine preamble expansion with resource-allocation waiting to mitigate access failures~\cite{lee2021enhanced}; Althumali\etal adapt service-class preamble partitions according to traffic load~\cite{althumali2022priority}; and Chowdhury and De prioritize ACB using queue occupancy~\cite{chowdhury2022queue}. Gedikli\etal further use deep reinforcement learning (DRL) to allocate preamble subsets among service classes~\cite{gedikli2022flexible_preamble}, while Piao and Lee jointly configure 2RA/4RA resources for heterogeneous IoT traffic~\cite{piao2024integrated_ra}. Together, these studies establish dynamic RACH resource partitioning as an effective means of accommodating heterogeneous access demands.

Beyond direct MAC-layer resource partitioning, access performance has also been improved through signaling, receiver processing, and alternative access structures. Jiang\etal introduce super-preambles for grant-free massive-MIMO access~\cite{jiang2019multiple}; Jang\etal classify collided preambles and timing-advance values using deep neural networks~\cite{jang2021cellular_ra_framework}; He and Ren exploit access-point clustering for pilot-collision resolution~\cite{he2022cluster}; and Khairy\etal optimize non-orthogonal multiple access (NOMA) transmission probabilities in dense multi-cell IoT networks~\cite{khairy2022data}. Liva and Polyanskiy survey unsourced multiple access as a coding-oriented framework for massive random access~\cite{liva2024unsourced}. More recent studies combine NOMA, 2RA, ACB, and contextual bandits for low-latency access~\cite{nie2025noma_rach}, employ reconfigurable-intelligent-surface-assisted grant-free transmission opportunities~\cite{marinello2025grantfree_ris}, or learn delay-oriented sensing-free access strategies~\cite{zhang2026delay_optimal_ra}. The present study addresses adaptive allocation of a finite preamble budget for mixed 2RA/4RA CBRA.

The time-varying nature of traffic, contention, and backlog further motivates learning-based RACH control. Tello-Oquendo\etal learn the ACB barring rate under mixed H2H/M2M traffic~\cite{tello2018rl_acb}; Pacheco-Paramo\etal apply double deep reinforcement learning to dynamic ACB~\cite{pacheco2019deep} and subsequently extend the control space to include random-access-opportunity periodicity for delay management~\cite{pacheco2020delay}. Bui and Pham jointly tune the barring factor and barring time with a dueling deep Q-network~\cite{bui2020energy_acb}. Jiang\etal coordinate learned ACB, backoff, and distributed-queuing controls~\cite{jiang2021decoupled_ra}, whereas Bai\etal employ a branching actor--critic for intelligent preamble selection~\cite{bai2021multiagent_ra}. Fan\etal combine hierarchical ACB/backoff with multi-agent DRL for prioritized access~\cite{fan2024joint_delay_energy}; Liu\etal apply a deep Q-network (DQN) to dynamic H2H/M2M RACH allocation~\cite{liu2024enhanced}; and Elmeligy\etal optimize priority-aware preamble-selection probabilities using a multi-armed bandit~\cite{elmeligy2025preamble_mab}. These studies establish access feedback as a basis for online adaptation, with control variables centered primarily on barring, access timing, and preamble selection.

The learning architecture must additionally accommodate incomplete network information and multidimensional discrete decisions. Mnih\etal establish DQN as a deep value-learning framework~\cite{mnih2015human}; Van Hasselt\etal introduce Double DQN (DDQN) to mitigate value overestimation~\cite{vanhasselt2016deep}; and Wang\etal develop dueling state-value/action-advantage decomposition~\cite{wang2016dueling}. Building on these foundations, Wang\etal formulate dynamic multichannel access as a partially observable Markov decision process and adapt DQN to time-varying channel dynamics~\cite{wang2018deep}. Tavakoli\etal address multidimensional discrete control through action branching with per-dimension value heads~\cite{tavakoli2018action_branching}. Nisioti and Thomos formulate adaptive irregular-repetition slotted ALOHA as a decentralized partially observable Markov decision process and accelerate Q-learning through virtual experience~\cite{nisioti2019fast}. Du\etal subsequently survey learning-based massive access and resource allocation as components of intelligent wireless-network management~\cite{du2020machine}, while Sohaib\etal apply a branching dueling Q-network to dynamic multichannel random access~\cite{sohaib2022dynamic}. These developments provide the methodological basis for observation-driven, scalable value learning over structured access decisions.

Taken together, the literature identifies three requirements pertinent to the present problem: adaptive allocation of a finite RACH resource, learning from aggregate access feedback, and scalable representation of multidimensional discrete decisions. \method{} addresses these requirements by associating one branch with each of four preamble pools defined by quality-of-service (QoS) class---delay-sensitive (DS) or delay-tolerant (DT)---and access mode---2RA or 4RA---and by projecting the selected branch multipliers to a feasible integer quota vector. Branch-wise value learning supports action selection, quota projection enforces the executable preamble budget, and subsequent access outcomes provide feedback for the next decision interval. The resulting architecture therefore couples the value-function representation directly to the controllable RACH allocation dimensions.

This paper makes the following contributions.
\begin{itemize}
	\item A finite-budget preamble-slicing formulation is established for mixed 2RA/4RA CBRA with four controllable QoS/access-mode pools.
	\item \method{} couples branch-wise value decisions with deterministic quota projection, producing nonnegative integer preamble allocations that preserve the CBRA budget.
	\item A procedure-level contention-resolution model is paired with explicit numerator/denominator definitions for success, collision, fallback, blocking, and successful-access delay.
	\item Cross-load method comparison, six-seed sensitivity analysis for \method{}, and targeted ablations characterize performance and design sensitivity under a common simulation framework.
\end{itemize}

The remainder of this paper is organized as follows.
Section~\ref{sec:system_problem} formulates the system and control problem,
Section~\ref{sec:method} presents \method{},
Section~\ref{sec:evaluation} reports the simulation evaluation, and
Section~\ref{sec:conclusion} summarizes the findings and limitations.

\section{System Model and Problem Formulation}\label{sec:system_problem}

\begin{table*}[!t]\rmfamily\scriptsize
	\renewcommand{\arraystretch}{1.03}
	\centering
	\caption{Main Notation Used Throughout the Manuscript}
	\label{tab:notation}
	\begin{tabularx}{\textwidth}{@{}lX@{\qquad}lX@{}}
		\toprule
		\textbf{Notation}              & \textbf{Description}                    & \textbf{Notation}               & \textbf{Description}              \\
		\midrule
		$\Nc$                          & Number of CBRA preambles                & $\Gset$                         & QoS/access-mode preamble-control groups    \\
		$\Bset$                        & Branch set of \method{}                 & $\Wset$                         & Candidate branch multiplier set       \\
		$\mathcal{U}_t$                & Active UE set at decision interval $t$  & $\mathcal{N}_t$                 & New-arrival UE set                \\
		$\mathcal{L}_t$                & Backlog carried into decision interval $t$        & $\mathcal{I}_t$                 & ACB-allowed attempted UE set      \\
		$A_{u,t}$                      & ACB allow indicator of UE $u$           & $p_t$                           & Exogenous ACB allow probability             \\
		$m_{u,t}$                      & Current access mode of UE $u$           & $B_{u,t}$                       & Four-step blocking indicator      \\
		$F_2,F_4$                      & Failure thresholds for 2RA fallback and 4RA blocking & $g(u,t)$                        & QoS/access-mode group of UE $u$            \\
		$N^{\mathrm{act}}_t$           & Active UE count                         & $N^{\mathrm{att}}_t$            & Attempted UE count                \\
		$N^{\mathrm{succ}}_t$          & Post-CR successful-UE count             & $N^{\mathrm{coll}}_t$           & Post-CR unresolved-UE count       \\
		$N^{\mathrm{fb}}_t$            & Two-step fallback count                 & $N^{\mathrm{blk}}_t$            & Four-step blocking count          \\
		$N^{\mathrm{res}}_t$           & UEs resolved from contended buckets     & $N^{\mathrm{unres}}_t$          & UEs unresolved in contended buckets \\
		$\qvec_t$                      & Integer preamble quota vector           & $\avec_t$                       & Branch action vector              \\
		$n_t^{\mathrm{rem}}$           & Residual preamble count after flooring  & $i$                             & Preamble index                    \\
		$\psi_{u,t}$                   & Preamble selected by UE $u$             & $\pi$                           & Observation-based control policy                    \\
		$\boldsymbol{\rho}$            & Baseline preamble-share vector          & $\Pi_{\Nc}(\cdot)$              & Quota-projection operator         \\
		$\svec_t$                      & Aggregate controller observation        & $r_t$                           & Scalar reward                     \\
		$\boldsymbol{\chi}_t$          & Latent UE-level simulator state         & $\gamma$                        & Discount factor                   \\
		$\theta,\theta^-$              & Online and target network parameters    & $\Dset$                         & Replay memory                     \\
		$\widehat{\eta}_{\mathrm{s}}$  & Success per attempt                     & $\widehat{\eta}_{\mathrm{c}}$   & Collision per attempt             \\
		$\widehat{\eta}_{\mathrm{fb}}$ & Fallback per active 2RA UE              & $\widehat{\eta}_{\mathrm{blk}}$ & Blocking per active 4RA UE        \\
		$\widehat{D}_{\mathrm{succ}}$  & Delay per successful UE                 & $\Tset_{\mathrm{tail}}$         & Tail set of epoch--interval pairs           \\
		$\Lambda$                      & Evaluated arrival-load set              & $\mathcal{Z}$                   & Seed set for proposed-method diagnostics \\
		$M,t_{\mathrm{warm}}$          & Mini-batch size and global warm-up boundary & $O$                         & Value-network output dimension    \\
		$n_{\mathrm{step}},n_{\mathrm{upd}}$ & Interaction and update counters       & $U$                            & Target-network update period      \\
		\bottomrule
	\end{tabularx}
\end{table*}

This section develops the CBRA system model, procedure-level contention accounting, performance measures, and the partially observed preamble-slicing problem. Table~\ref{tab:notation} summarizes the main notation, and Fig.~\ref{fig:scenario} illustrates the scenario interpretation and access process.

\subsection{Network, Traffic, and Decision-Interval Model}

Consider one gNB serving H2H and M2M UEs over $\Nc$ CBRA preambles.
Each simulation epoch comprises $T$ RACH decision intervals indexed by
$t\in\{1,\ldots,T\}$. The epoch index $e$ is suppressed in the
interval-level model and included explicitly in cross-epoch aggregation.
The active UE set in an interval is
\begin{equation}
	\mathcal{U}_t=\mathcal{N}_t\cup\mathcal{L}_t,
	\qquad
	\mathcal{N}_t\cap\mathcal{L}_t=\varnothing,
	\label{eq:active_set}
\end{equation}
where $\mathcal{N}_t$ contains new arrivals and $\mathcal{L}_t$ contains UEs carried over from preceding intervals.

Each UE $u$ is characterized by a QoS class $c_u\in\{\mathrm{DS},\mathrm{DT}\}$, a user type $h_u\in\{\mathrm{H2H},\mathrm{M2M}\}$, and a current access mode $m_{u,t}\in\{2\mathrm{RA},4\mathrm{RA}\}$. The user-type and QoS labels are distinct model attributes; their evaluated joint distribution is specified in Section~\ref{sec:evaluation}. The proposed controller allocates preambles over four QoS/access-mode groups
\begin{equation}
	\Gset=
	\{\mathrm{DS2},\mathrm{DT2},\mathrm{DS4},\mathrm{DT4}\},
	\label{eq:group_set}
\end{equation}
with
\begin{equation}
	g(u,t)=
	\begin{cases}
		\mathrm{DS2}, & c_u=\mathrm{DS},\;m_{u,t}=2\mathrm{RA},\\
		\mathrm{DT2}, & c_u=\mathrm{DT},\;m_{u,t}=2\mathrm{RA},\\
		\mathrm{DS4}, & c_u=\mathrm{DS},\;m_{u,t}=4\mathrm{RA},\\
		\mathrm{DT4}, & c_u=\mathrm{DT},\;m_{u,t}=4\mathrm{RA}.
	\end{cases}
	\label{eq:group_map}
\end{equation}
The proposed controller assigns quotas by QoS and current access mode;
user type enters the traffic model and subgroup statistics.

Fig.~\ref{fig:scenario} summarizes the modeled access sequence from arrivals and ACB gating to preamble contention, procedure-level resolution, fallback, blocking, and backlog carryover.
\begin{figure*}[!t]
	\centering
	\includegraphics[width=\textwidth]{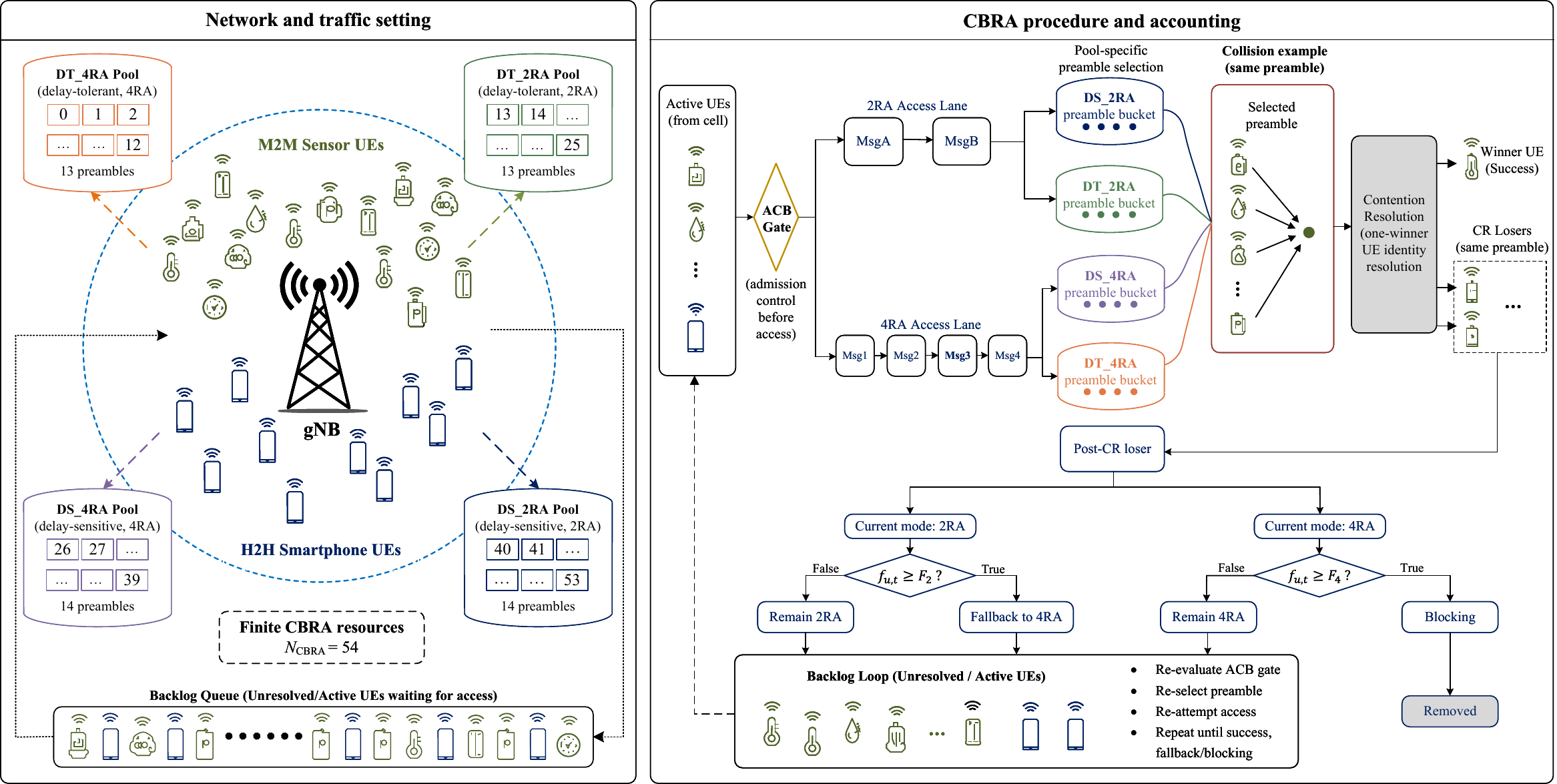}
	\caption{System model for QoS-aware CBRA preamble slicing with H2H/M2M
		arrivals, ACB gating, 2RA/4RA access, procedure-level contention
		resolution, fallback, blocking, and backlog carryover. The displayed
		13/14/14/13 quotas illustrate one feasible preamble allocation.}
	\label{fig:scenario}
\end{figure*}
A gNB-side controller updates the four preamble quotas once per decision interval using aggregate traffic, preceding quota allocation, recent access outcomes, and backlog information.

\begin{assumption}[Controller information]
	Before current-interval arrivals are introduced, the controller receives
	pending-UE counts by QoS/access mode and by access mode/user type,
	preceding quota shares, the nominal arrival load, and recent access
	outcomes. Action selection precedes the introduction of new UEs.
	\label{assump:observability}
\end{assumption}

\subsection{ACB Gate and Attempt Process}

Before preamble selection, each active UE passes through an ACB gate. Let
$A_{u,t}\in\{0,1\}$ denote the access decision, with
\begin{equation}
	\Pr(A_{u,t}=1)=p_t,
	\qquad 0\le p_t\le1,
	\label{eq:acb_gate}
\end{equation}
where $p_t$ is an exogenous ACB allow probability. The set of UEs admitted to
the current access attempt is
\begin{equation}
	\mathcal{I}_t=
	\{u\in\mathcal{U}_t:A_{u,t}=1\},
	\label{eq:attempt_set}
\end{equation}
and
\begin{equation}
	N_t^{\mathrm{act}}=|\mathcal{U}_t|,
	\qquad
	N_t^{\mathrm{att}}=|\mathcal{I}_t|.
	\label{eq:active_attempt_count}
\end{equation}
The evaluated value of $p_t$ is specified with the simulation configuration in
Section~\ref{sec:evaluation}.

\subsection{Two-Step and Four-Step Procedure Dynamics}

The simulator assigns each arriving UE an initial access mode according to the traffic-generation rule specified in Section~\ref{sec:evaluation}.
For overhead accounting, each 2RA attempt is assigned a nominal message
cost of two and each 4RA attempt a cost of four:
\begin{equation}
	C^{\mathrm{msg}}_{u,t}=
	\begin{cases}
		2, & m_{u,t}=2\mathrm{RA},\\
		4, & m_{u,t}=4\mathrm{RA}.
	\end{cases}
	\label{eq:msg_cost_user}
\end{equation}
The aggregate nominal message cost in interval $t$ is
\begin{equation}
	C^{\mathrm{msg}}_t
	=
	\sum_{u\in\mathcal{I}_t}C^{\mathrm{msg}}_{u,t}.
	\label{eq:msg_cost_total}
\end{equation}
This quantity summarizes nominal procedure overhead for the diagnostic analysis.

Each UE maintains a per-mode counter initialized to zero at arrival.
The counter increases once per unsuccessful decision interval, including
ACB denial and unresolved preamble contention, and resets upon successful
resolution or 2RA fallback. In the threshold tests below, $f_{u,t}$ denotes
the value after the current outcome is recorded and before any fallback
reset. The mode and blocking updates are
\begin{align}
	m_{u,t+1}
	&=
	\begin{cases}
		4\mathrm{RA}, & m_{u,t}=2\mathrm{RA},\;f_{u,t}\ge F_2,\\
		m_{u,t},      & \text{otherwise},
	\end{cases}
	\label{eq:fallback_mode}\\
	B_{u,t+1}
	&=
	\mathbf{1}
	\{m_{u,t}=4\mathrm{RA},\;f_{u,t}\ge F_4\},
	\label{eq:blocking_indicator}
\end{align}
where $F_2$ and $F_4$ are the per-mode unsuccessful-interval thresholds.
The resulting event counts are
\begin{align}
	N_t^{\mathrm{fb}}
	&=
	\left|
	\{u\in\mathcal{U}_t:
	m_{u,t}=2\mathrm{RA},\;
	m_{u,t+1}=4\mathrm{RA}\}
	\right|,
	\label{eq:fallback_count}\\
	N_t^{\mathrm{blk}}
	&=
	\sum_{u\in\mathcal{U}_t}B_{u,t+1}.
	\label{eq:blocking_count}
\end{align}
Successfully resolved and blocked UEs leave the active population; all remaining UEs form the next backlog:
\begin{equation}
	\mathcal{L}_{t+1}
	=
	\{u\in\mathcal{U}_t:
	u\notin\mathcal{S}^{\mathrm{succ}}_t,\;
	B_{u,t+1}=0\},
	\label{eq:backlog_transition}
\end{equation}
where $\mathcal{S}^{\mathrm{succ}}_t$ is the set of successfully resolved UEs.

\subsection{Procedure-Level Contention-Resolution Accounting}

For a nonempty quota pool, each admitted UE independently selects one
preamble uniformly from that pool. Using disjoint pool indices within
$\{1,\ldots,\Nc\}$, let $\psi_{u,t}$ denote the selected preamble and define
\begin{equation}
	\mathcal{P}_{t,i}
	=
	\{u\in\mathcal{I}_t:\psi_{u,t}=i\}.
	\label{eq:preamble_bucket}
\end{equation}

\begin{assumption}[Contention-resolution accounting]
	A singleton preamble bucket contributes one resolved UE. For each
	non-singleton bucket, one UE is recorded as resolved and the remaining UEs are
	recorded as unresolved for the current attempt. The accounting operates at the
	random-access procedure level; physical-layer capture and multi-user decoding
	are not modeled.
	\label{assump:procedure_resolution}
\end{assumption}

Let
\begin{align}
	N_t^{\mathrm{sing}}
	&=
	\left|
	\{i:|\mathcal{P}_{t,i}|=1\}
	\right|,
	\label{eq:singleton_count}\\
	N_t^{\mathrm{cb}}
	&=
	\left|
	\{i:|\mathcal{P}_{t,i}|\ge2\}
	\right|.
	\label{eq:contended_bucket_count}
\end{align}
Under Assumption~\ref{assump:procedure_resolution},
\begin{align}
	N_t^{\mathrm{res}}
	&=N_t^{\mathrm{cb}},
	\label{eq:resolved_contended_count}\\
	N_t^{\mathrm{unres}}
	&=
	\sum_{i:|\mathcal{P}_{t,i}|\ge2}
	\left(|\mathcal{P}_{t,i}|-1\right).
	\label{eq:unresolved_contended_count}
\end{align}
The success and collision counts used throughout the manuscript are therefore
\begin{align}
	N_t^{\mathrm{succ}}
	&=
	N_t^{\mathrm{sing}}+N_t^{\mathrm{res}},
	\label{eq:success_accounting}\\
	N_t^{\mathrm{coll}}
	&=
	N_t^{\mathrm{unres}}.
	\label{eq:collision_accounting}
\end{align}
Here, $N_t^{\mathrm{coll}}$ denotes UEs remaining unresolved after the procedure-level contention accounting.
When every attempted group has a nonempty pool, each attempted UE is
recorded once as resolved or unresolved, giving
\begin{equation}
	N_t^{\mathrm{att}}
	=
	N_t^{\mathrm{succ}}
	+
	N_t^{\mathrm{coll}}.
	\label{eq:post_cr_identity}
\end{equation}

\subsection{Performance Metrics and Aggregation}\label{sec:kpi}

For an analysis set $\Tset$ of epoch--interval pairs $(e,t)$, define
\begin{equation}
	\widehat{R}_{x/y}(\Tset)
	=
	\frac{\sum_{(e,t)\in\Tset}x_{e,t}}
	{\sum_{(e,t)\in\Tset}y_{e,t}},
	\label{eq:ratio_of_sums}
\end{equation}
for nonnegative quantities with a positive total denominator.
Numerators and denominators are pooled separately across the selected
intervals. Interval-level ratios used in observations and rewards are
assigned zero when their denominators vanish.

\subsubsection{Primary Metrics}

Let
\begin{equation}
	N_t^{\mathrm{act},m}
	=
	\left|
	\{u\in\mathcal{U}_t:m_{u,t}=m\}
	\right|,
	\qquad
	m\in\{2\mathrm{RA},4\mathrm{RA}\},
	\label{eq:active_mode_count}
\end{equation}
Let $\tau_u$ be the arrival interval of UE $u$. Successful resolution in
interval $t$ gives $d_{u,t}=t-\tau_u$, measured in decision intervals.
Success in the arrival interval has zero delay, and waiting during ACB
denial contributes to the elapsed delay. Mode-specific delay summaries
group successful UEs by initial access mode.
The successful-user delay sum is
\begin{equation}
	D_t^{\mathrm{succ}}
	=
	\sum_{u\in\mathcal{S}^{\mathrm{succ}}_t}d_{u,t}.
	\label{eq:delay_success_sum}
\end{equation}
The primary metrics are
\begin{subequations}\label{eq:primary_metric_estimators}
	\begin{align}
		\widehat{\eta}_{\mathrm{s}}(\Tset)
		&=
		\widehat{R}_{N^{\mathrm{succ}}/N^{\mathrm{att}}}(\Tset),
		\label{eq:success_per_attempt}\\
		\widehat{\eta}_{\mathrm{c}}(\Tset)
		&=
		\widehat{R}_{N^{\mathrm{coll}}/N^{\mathrm{att}}}(\Tset),
		\label{eq:collision_per_attempt}\\
		\widehat{\eta}_{\mathrm{fb}}(\Tset)
		&=
		\widehat{R}_{N^{\mathrm{fb}}/
			N^{\mathrm{act},2\mathrm{RA}}}(\Tset),
		\label{eq:fallback_per_2ra}\\
		\widehat{\eta}_{\mathrm{blk}}(\Tset)
		&=
		\widehat{R}_{N^{\mathrm{blk}}/
			N^{\mathrm{act},4\mathrm{RA}}}(\Tset),
		\label{eq:blocking_per_4ra}\\
		\widehat{D}_{\mathrm{succ}}(\Tset)
		&=
		\widehat{R}_{D^{\mathrm{succ}}/N^{\mathrm{succ}}}(\Tset).
		\label{eq:delay_per_success}
	\end{align}
\end{subequations}
Equation~\eqref{eq:post_cr_identity} further gives
\begin{equation}
	\widehat{\eta}_{\mathrm{s}}(\Tset)
	+
	\widehat{\eta}_{\mathrm{c}}(\Tset)
	=1
\end{equation}
for horizons containing at least one attempted UE and nonempty pools for
all attempted groups. The evaluated configurations satisfy this condition.

\subsubsection{Diagnostic Metrics}

Active and attempted access pressure per preamble are
\begin{align}
	\widehat{\rho}_{\mathrm{act}}(\Tset)
	&=
	\frac{\sum_{(e,t)\in\Tset}N_{e,t}^{\mathrm{act}}}
	{|\Tset|\Nc},
	\label{eq:active_per_preamble}\\
	\widehat{\rho}_{\mathrm{att}}(\Tset)
	&=
	\frac{\sum_{(e,t)\in\Tset}N_{e,t}^{\mathrm{att}}}
	{|\Tset|\Nc}.
	\label{eq:attempt_per_preamble}
\end{align}
Let $N_t^{\mathrm{cont,UE}}$ denote the number of attempted UEs belonging to non-singleton preamble buckets before procedure-level resolution.
The corresponding contention pressure is
\begin{equation}
	\widehat{\eta}_{\mathrm{cont}}(\Tset)
	=
	\widehat{R}_{N^{\mathrm{cont,UE}}/N^{\mathrm{att}}}(\Tset).
	\label{eq:contention_pressure}
\end{equation}
Mean post-interval backlog over the tail window is
\begin{equation}
	\widehat{L}_{\mathrm{tail}}
	=
	\frac{1}{|\Tset_{\mathrm{tail}}|}
	\sum_{(e,t)\in\Tset_{\mathrm{tail}}}
	|\mathcal{L}_{e,t+1}|.
	\label{eq:backlog_tail}
\end{equation}
Here $\mathcal{L}_{e,t+1}$ is the backlog after interval $t$ of epoch $e$;
$\mathcal{L}_{e,T+1}$ records the final-interval backlog before the next
epoch resets.
Nominal message cost is reported per attempt and per successful UE:
\begin{align}
	\widehat{C}^{\mathrm{att}}(\Tset)
	&=
	\widehat{R}_{C^{\mathrm{msg}}/N^{\mathrm{att}}}(\Tset),
	\label{eq:msg_pressure_attempt}\\
	\widehat{C}^{\mathrm{succ}}(\Tset)
	&=
	\widehat{R}_{C^{\mathrm{msg}}/N^{\mathrm{succ}}}(\Tset).
	\label{eq:msg_pressure_success}
\end{align}
These quantities are used as secondary diagnostics for access pressure, backlog, contention, and signaling overhead.

\subsection{Partially Observed Preamble-Slicing Problem}\label{sec:problem}

\subsubsection{Controller Observation}

Let $\boldsymbol{\chi}_t$ denote the latent UE-level simulator state containing the variables required for access-mode evolution, failure histories, realized contention, and backlog evolution.
Under Assumption~\ref{assump:observability}, the controller receives only the aggregate observation $\svec_t$. The policy is therefore written as
\begin{equation}
	\pi(\avec_t\mid\svec_t),
	\label{eq:observation_policy}
\end{equation}
with no requirement that $\svec_t$ constitute a complete Markov state.

Let $\lambda$ denote the nominal number of new UEs per interval and set
\begin{equation}
	\begin{aligned}
		H_\lambda&=\max\{\Nc,5\lambda,1\},\\
		x_{t,g}^{\mathrm{quota}}&=q_{t,g}^{\mathrm{prev}}/\Nc.
	\end{aligned}
	\label{eq:observation_group_norm}
\end{equation}
The pending subsets $\mathcal{L}_{t,g}$ and $\mathcal{L}_{t,m,h}$ group
UEs by QoS/current access mode and by current access mode/user type,
respectively. Their normalized counts form
\begin{equation}
	\begin{aligned}
		\xvec_t^{\mathrm{pend}}=\big[&
		(|\mathcal{L}_{t,g}|/H_\lambda)_{g\in\Gset},\\
		&(|\mathcal{L}_{t,m,h}|/H_\lambda)_{m,h},
		|\mathcal{L}_t|/H_\lambda\big].
	\end{aligned}
	\label{eq:observation_backlog}
\end{equation}
Group entries follow DS2, DT2, DS4, DT4. Mode/type entries follow
2RA--H2H, 2RA--M2M, 4RA--H2H, 4RA--M2M.
The preceding global outcomes are
\begin{equation}
	\begin{aligned}
		\xvec_t^{\mathrm{perf}}=\big[&
		\eta_{\mathrm{s},t-1}^{\mathrm{act}},
		\eta_{\mathrm{c},t-1},\eta_{\mathrm{blk},t-1},\\
		&\eta_{\mathrm{fb},t-1},\bar{D}_{t-1}/10\big],
	\end{aligned}
	\label{eq:observation_perf}
\end{equation}
where $\eta_{\mathrm{s},t}^{\mathrm{act}}
=N_t^{\mathrm{succ}}/\max(N_t^{\mathrm{act}},1)$ and
$\bar{D}_t=D_t^{\mathrm{succ}}/\max(N_t^{\mathrm{succ}},1)$.
The other rate entries use their interval-level denominators from
Section~\ref{sec:kpi}. DS-specific outcomes form
\begin{equation}
	\begin{aligned}
		\xvec_t^{\mathrm{DS}}=\big[&
		(\eta_{\mathrm{blk},t-1}^{\mathrm{DS4},h})_h,
		(\eta_{\mathrm{fb},t-1}^{\mathrm{DS2},h})_h,\\
		&(\bar{D}_{t-1}^{\mathrm{DS4},h}/10)_h,
		(\bar{D}_{t-1}^{\mathrm{DS2},h}/10)_h\big],
	\end{aligned}
	\label{eq:observation_ds}
\end{equation}
with $h$ ordered as H2H, M2M. Subgroup rates use current-mode active
populations; subgroup delays use successful UEs grouped by initial mode.
The flattened observation is
\begin{equation}
	\begin{aligned}
		\svec_t=\big[&
		\xvec_t^{\mathrm{pend}},\xvec_t^{\mathrm{quota}},
		\lambda/100,\\
		&p_t^{\mathrm{prev}},\xvec_t^{\mathrm{perf}},
		\xvec_t^{\mathrm{DS}}\big]\in\mathbb{R}^{28},
	\end{aligned}
	\label{eq:observation_concat}
\end{equation}
where $\xvec_t^{\mathrm{quota}}=(x_{t,g}^{\mathrm{quota}})_{g\in\Gset}$
and $p_t^{\mathrm{prev}}$ is the preceding logged ACB probability.
The six blocks contain $9$, $4$, $1$, $1$, $5$, and $8$ entries,
respectively. At epoch initialization, pending counts and preceding
metric entries, including the previous-ACB feature, are zero; quota
features reflect the static allocation.

\subsubsection{Four-Branch Action Space}

The branch set follows the executable quota dimensions:
\begin{equation}
	\Bset=\Gset.
	\label{eq:branch_set}
\end{equation}
Each branch selects one multiplier from
\begin{equation}
	\Wset=\{0.50,0.75,1.00,1.25,1.50\},
	\label{eq:branch_values}
\end{equation}
giving
\begin{equation}
	\avec_t=(a_{t,b})_{b\in\Bset}\in
	\Aset=\Wset^{|\Bset|}.
	\label{eq:branch_action}
\end{equation}
The four branches therefore define
\begin{equation}
	N_{\mathrm{tuple}}^{(4)}
	=
	|\Wset|^{|\Bset|}
	=
	5^4
	=
	625
	\label{eq:branch_tuple_count}
\end{equation}
pre-projection branch-action tuples.

\subsubsection{Quota Projection}

Let $\boldsymbol{\rho}=(\rho_b)_{b\in\Bset}$ denote the fixed baseline preamble-share vector, with
\begin{equation}
	\rho_b\ge0,
	\qquad
	\sum_{b\in\Bset}\rho_b=1.
	\label{eq:baseline_share}
\end{equation}
For branch $b$, the selected multiplier produces the score
\begin{equation}
	z_{t,b}=a_{t,b}\rho_b.
	\label{eq:branch_score}
\end{equation}
Since all elements of $\Wset$ are positive, $\sum_{b\in\Bset}z_{t,b}>0$, and the normalized real-valued quota is
\begin{equation}
	\tilde{q}_{t,b}
	=
	\Nc
	\frac{z_{t,b}}
	{\sum_{j\in\Bset}z_{t,j}}.
	\label{eq:unrounded_quota}
\end{equation}
Initial integer quotas are obtained by flooring:
\begin{equation}
	q_{t,b}^{(0)}
	=
	\left\lfloor
	\tilde{q}_{t,b}
	\right\rfloor,
	\label{eq:floor_quota}
\end{equation}
leaving
\begin{equation}
	n_t^{\mathrm{rem}}
	=
	\Nc
	-
	\sum_{b\in\Bset}
	q_{t,b}^{(0)}
	\label{eq:quota_residual}
\end{equation}
preambles to be assigned.

Let $\prec_{\Bset}$ order the branches as DS2, DT2, DS4, DT4.
For $\ell=1,\ldots,n_t^{\mathrm{rem}}$, the residual allocation selects
\begin{equation}
	b_\ell^\star
	\in
	\underset{b\in\Bset:z_{t,b}>0}{\arg\min}
	\frac{q_{t,b}^{(\ell-1)}}{z_{t,b}},
	\label{eq:residual_ratio_choice}
\end{equation}
with ties resolved by $\prec_{\Bset}$, and updates
\begin{equation}
	q_{t,b}^{(\ell)}
	=
	q_{t,b}^{(\ell-1)}
	+
	\mathbf{1}\{b=b_\ell^\star\}.
	\label{eq:residual_ratio_update}
\end{equation}
The final allocation is
\begin{equation}
	q_{t,b}
	=
	q_{t,b}^{(n_t^{\mathrm{rem}})},
	\qquad b\in\Bset,
	\label{eq:final_quota}
\end{equation}
or, equivalently,
\begin{equation}
	\qvec_t
	=
	\Pi_{\Nc}
	(\avec_t,\boldsymbol{\rho}).
	\label{eq:projection_operator}
\end{equation}

\begin{proposition}[Quota feasibility]
	The projection
	$\Pi_{\Nc}$ returns an integer quota vector satisfying
	\begin{equation}
		q_{t,b}\in\mathbb{Z}_{\ge0},
		\qquad
		\sum_{b\in\Bset}q_{t,b}=\Nc.
		\label{eq:quota_feasible}
	\end{equation}
\end{proposition}

\begin{proof}
	Flooring produces nonnegative integer quotas whose sum is at most $\Nc$.
	The subsequent $n_t^{\mathrm{rem}}$ single-unit assignments increase the
	aggregate quota by exactly $n_t^{\mathrm{rem}}$, yielding a final sum of
	$\Nc$.
\end{proof}

Under the proposed four-pool formulation, each branch controls one
QoS/access-mode quota.

\subsubsection{Reward and Objective}

The scalar reward $r_t$ is a dimensionless training score combining access success, collision, fallback, blocking, delay, and DS-specific penalties:
\begin{align}
	r_t
	&=
	w_{\mathrm{s4}}\eta_{\mathrm{s4},t}
	+w_{\mathrm{s2}}\eta_{\mathrm{s2},t}
	+w_{\mathrm{c}}\eta_{\mathrm{c},t}
	\notag\\
	&\quad
	+w_{\mathrm{blk}}\eta_{\mathrm{blk},t}
	+w_{\mathrm{fb}}\eta_{\mathrm{fb},t}
	+w_{\mathrm{d4}}\frac{\bar{D}_{4,t}}{10}
	\notag\\
	&\quad
	+w_{\mathrm{d2}}\frac{\bar{D}_{2,t}}{10}
	+w_{\mathrm{sb}}\delta^{\mathrm{sen}}_{\mathrm{blk},t}
	+w_{\mathrm{sf}}\delta^{\mathrm{sen}}_{\mathrm{fb},t}
	+w_{\mathrm{sd}}
	\frac{\delta^{\mathrm{sen}}_{\mathrm{d},t}}{10}.
	\label{eq:reward}
\end{align}
Here $\eta_{\mathrm{s2},t}$ and $\eta_{\mathrm{s4},t}$ use current-mode
active populations. The delay summaries $\bar{D}_{2,t}$ and
$\bar{D}_{4,t}$ group successful UEs by initial access mode. Collision,
fallback, and blocking use the interval-level counterparts of the
ratios in Section~\ref{sec:kpi}.
DS-specific rates use the corresponding current-mode active populations,
and DS-specific delays retain initial-mode grouping. The DS-specific terms are
\begin{subequations}\label{eq:sensitive_reward_terms}
	\begin{align}
		\delta^{\mathrm{sen}}_{\mathrm{blk},t}
		&=
		0.5\,
		\eta_{\mathrm{blk},t}^{\mathrm{DS4,H2H}}
		+
		0.5\,
		\eta_{\mathrm{blk},t}^{\mathrm{DS4,M2M}},
		\\
		\delta^{\mathrm{sen}}_{\mathrm{fb},t}
		&=
		0.5\,
		\eta_{\mathrm{fb},t}^{\mathrm{DS2,H2H}}
		+
		0.5\,
		\eta_{\mathrm{fb},t}^{\mathrm{DS2,M2M}},
		\\
		\delta^{\mathrm{sen}}_{\mathrm{d},t}
		&=
		0.3\,
		\bar{D}_{t}^{\mathrm{DS4,H2H}}
		+
		0.2\,
		\bar{D}_{t}^{\mathrm{DS4,M2M}}
		\notag\\
		&\quad
		+
		0.3\,
		\bar{D}_{t}^{\mathrm{DS2,H2H}}
		+
		0.2\,
		\bar{D}_{t}^{\mathrm{DS2,M2M}}.
	\end{align}
\end{subequations}

The default coefficient vector is
\begin{equation}
	\begin{aligned}
		\mathbf{w}^{\mathrm{default}}
		&=
		(1.50,1.50,-0.50,-0.60,-0.60,\\
		&\quad -0.25,-0.40,-1.00,-0.75,-0.40),
	\end{aligned}
	\label{eq:mild_penalty}
\end{equation}
ordered according to the terms in \eqref{eq:reward}. Delay terms are scaled by $1/10$ before weighting.
The coefficient vector is fixed across the main comparisons; the reward-profile ablation is specified in Section~\ref{sec:evaluation}.
Performance comparison uses the primary metrics in Section~\ref{sec:kpi}.

The controller seeks an observation-based policy maximizing discounted
return over one $T$-interval epoch:
\begin{subequations}\label{eq:optimization}
	\begin{align}
		\max_{\pi}\quad
		&J(\pi)
		=
		\E_{\pi}
		\left[
		\sum_{t=1}^{T}
		\gamma^{t-1} r_t
		\right]
		\label{eq:objective}\\
		\text{s.t.}\quad
		&a_{t,b}\in\Wset,
		\qquad
		\forall b\in\Bset,
		\label{eq:c_action}\\
		&q_{t,b}\in\mathbb{Z}_{\ge0},
		\qquad
		\forall b\in\Bset,
		\label{eq:c_quota_nonnegative}\\
		&\sum_{b\in\Bset}q_{t,b}=\Nc,
		\label{eq:c_quota_sum}\\
		&\qvec_t=
		\Pi_{\Nc}
		(\avec_t,\boldsymbol{\rho}).
		\label{eq:c_projection}
	\end{align}
\end{subequations}
The environment transition follows the ACB, access-mode, contention-resolution, and backlog dynamics defined in the preceding subsections.

\section{The \method{} Controller}\label{sec:method}

\begin{figure*}[!t]
	\centering
	\includegraphics[width=\textwidth]{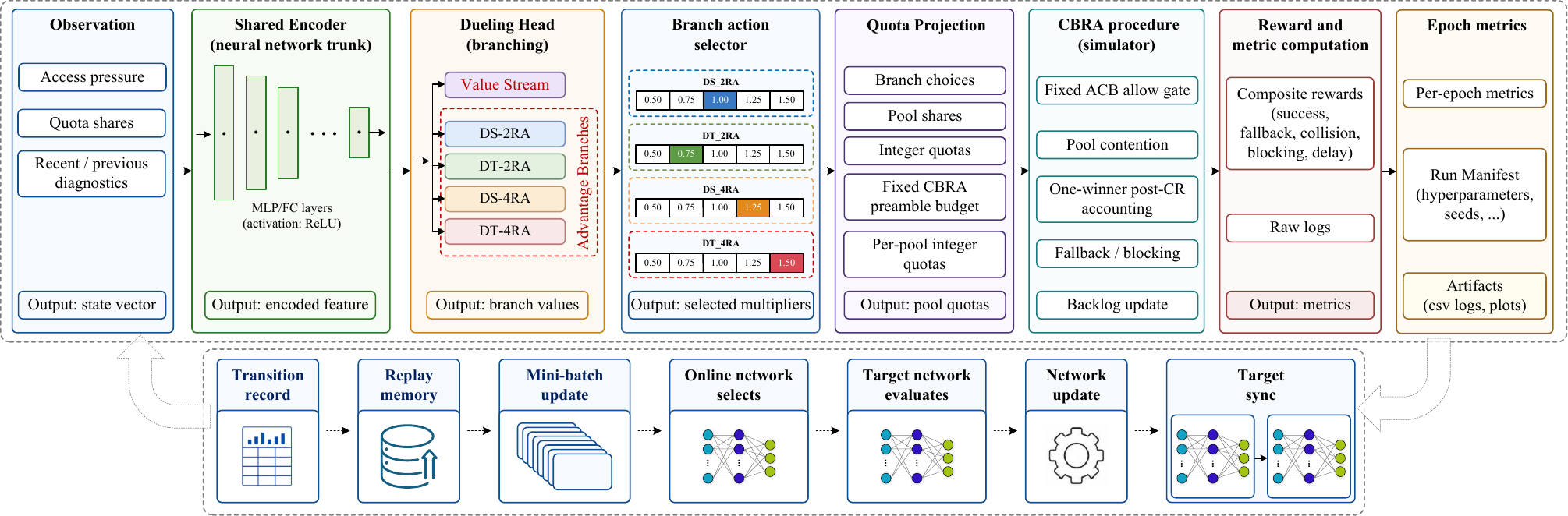}
	\caption{Control and learning architecture of \method{}. A shared encoder maps the aggregate observation to branch-wise action values; selected branch multipliers are projected to a feasible integer quota vector; and subsequent CBRA outcomes provide reward and next-observation feedback for replay-based Double DQN updates.}
	\label{fig:method}
\end{figure*}

\method{} maps the aggregate observation $\svec_t$ to one multiplier decision for each of the four executable preamble-quota dimensions defined in Section~\ref{sec:problem}.
As illustrated in Fig.~\ref{fig:method}, a shared encoder and branch-wise dueling value heads support action selection, while the deterministic operator $\Pi_{\Nc}$ converts the selected multipliers to a budget-feasible integer quota vector.
The resulting CBRA outcomes provide the reward and next observation used for replay-based learning.
This section specifies the branch-wise value representation, action selection, Double DQN update, training procedure, and output dimensionality.

\subsection{Branching Dueling Value Representation and Action Selection}

The controller processes $\svec_t$ with two fully connected hidden layers
and rectified linear unit (ReLU) activations. Each layer has width $128$
in the default configuration, and Adam minimizes the mini-batch loss.
The scalar value/action-advantage decomposition follows the dueling architecture of Wang\etal~\cite{wang2016dueling}.
The shared representation with one action branch per decision dimension follows the action-branching architecture introduced by Tavakoli\etal~\cite{tavakoli2018action_branching}.

The shared encoder produces
\begin{equation}
	\hvec_t=f_\theta(\svec_t),
	\label{eq:encoder}
\end{equation}
from which the network estimates a scalar value $V_\theta(\svec_t)$ and a branch-specific advantage $A_{\theta,b}(\svec_t,a_b)$ for each $b\in\Bset$ and $a_b\in\Wset$.
The corresponding branch action value is
\begin{equation}
	Q_{\theta,b}(\svec_t,a_b)
	=
	V_\theta(\svec_t)
	+
	A_{\theta,b}(\svec_t,a_b)
	-
	\frac{1}{|\Wset|}
	\sum_{a'\in\Wset}
	A_{\theta,b}(\svec_t,a').
	\label{eq:dueling_branch}
\end{equation}

The greedy action of branch $b$ is
\begin{equation}
	a^\star_{t,b}
	=
	\underset{a_b\in\Wset}{\arg\max}\;
	Q_{\theta,b}(\svec_t,a_b).
	\label{eq:greedy_branch}
\end{equation}

A single exploration draw $\xi_t\sim\mathrm{Uniform}(0,1)$ determines whether the complete branch vector follows exploration or greedy selection.
Conditional on exploration, each branch independently samples one action $a^{\mathrm{rnd}}_{t,b}$ uniformly from $\Wset$.
The executed action vector is
\begin{equation}
	\avec_t
	=
	\begin{cases}
		(a^{\mathrm{rnd}}_{t,b})_{b\in\Bset},
		& \xi_t<\varepsilon,\\[1mm]
		(a^\star_{t,b})_{b\in\Bset},
		& \xi_t\ge\varepsilon.
	\end{cases}
	\label{eq:epsilon_greedy}
\end{equation}

The selected branch vector is converted to the executable allocation by
\begin{equation}
	\qvec_t
	=
	\Pi_{\Nc}(\avec_t,\boldsymbol{\rho}),
	\label{eq:method_projection}
\end{equation}
using the projection defined in Section~\ref{sec:problem}.
Thus, value estimation is factorized across the four decision dimensions, while quota execution remains coupled through the common preamble budget.

\subsection{Branch-Specific Double DQN Update}

Action selection and target evaluation follow the Double DQN separation introduced by Van Hasselt\etal~\cite{vanhasselt2016deep}.
Among the temporal-difference target constructions examined for action branching, \method{} uses the branch-specific form described by Tavakoli\etal~\cite{tavakoli2018action_branching}.

For a replay transition $(\svec_t,\avec_t,r_t,\svec_{t+1},d_t)$,
with $d_t=\mathbf{1}\{t=T\}$ marking the end of an epoch,
the online network first selects
\begin{equation}
	a^+_{t+1,b}
	=
	\underset{a_b\in\Wset}{\arg\max}\;
	Q_{\theta,b}(\svec_{t+1},a_b).
	\label{eq:double_select}
\end{equation}
The target network then evaluates the selected action:
\begin{equation}
	y_{t,b}
	=
	r_t
	+
	\gamma(1-d_t)
	Q_{\theta^-,b}
	(\svec_{t+1},a^+_{t+1,b}).
	\label{eq:double_target}
\end{equation}
The branch-wise temporal-difference error is
\begin{equation}
	\delta_{t,b}
	=
	y_{t,b}
	-
	Q_{\theta,b}(\svec_t,a_{t,b}),
	\label{eq:td_error}
\end{equation}
and the transition loss is
\begin{equation}
	L_t^{\mathrm{TD}}(\theta)
	=
	\frac{1}{|\Bset|}
	\sum_{b\in\Bset}
	\delta_{t,b}^{\,2}.
	\label{eq:loss}
\end{equation}
For mini-batch training, $L_t^{\mathrm{TD}}(\theta)$ is averaged over the sampled transitions.
Parameter updates begin after the global interaction count exceeds $t_{\mathrm{warm}}$ and the replay memory contains at least $M$ transitions.
The target parameters $\theta^-$ are synchronized with the online parameters every $U$ gradient updates.

Algorithm~\ref{alg:training} summarizes the interaction and learning cycle.
Online and target parameters, replay memory, and global counters persist
across epochs; the simulator is reset at each epoch boundary.
Raw metric numerators and denominators are logged during simulation and
aggregated in Section~\ref{sec:evaluation} according to
\eqref{eq:ratio_of_sums}.

\begin{algorithm}[!t]
	\caption{Training Procedure for \method{}}
	\label{alg:training}\scriptsize
	\begin{algorithmic}[1]
		\REQUIRE Training epochs $E$, decision intervals per epoch $T$, branch set $\Bset$, candidate set $\Wset$, baseline shares $\boldsymbol{\rho}$, total preambles $\Nc$, exploration probability $\varepsilon$, discount factor $\gamma$, mini-batch size $M$, warm-up boundary $t_{\mathrm{warm}}$, target-update period $U$
		\ENSURE Trained online-network parameters $\theta$
		\STATE Initialize $\theta$ and set $\theta^-\leftarrow\theta$.
		\STATE Initialize replay memory $\Dset$ and counters $n_{\mathrm{step}}\leftarrow0$, $n_{\mathrm{upd}}\leftarrow0$.
		\FOR{$e=1,\ldots,E$}
		\STATE Reset the simulator to an empty backlog, static quotas,
		and zero preceding metrics.
		\FOR{$t=1,\ldots,T$}
		\STATE Receive aggregate observation $\svec_t$.
		\STATE During replay warm-up, sample each branch action independently and uniformly from $\Wset$; thereafter, select $\avec_t$ using \eqref{eq:greedy_branch}--\eqref{eq:epsilon_greedy}.
		\STATE Compute $\qvec_t=\Pi_{\Nc}(\avec_t,\boldsymbol{\rho})$.
		\STATE Execute the CBRA decision interval under $\qvec_t$ and obtain $r_t$, $\svec_{t+1}$, and terminal indicator $d_t$.
		\STATE Log the raw metric numerators and denominators.
		\STATE Store $(\svec_t,\avec_t,r_t,\svec_{t+1},d_t)$ in $\Dset$.
		\STATE Set $n_{\mathrm{step}}\leftarrow n_{\mathrm{step}}+1$.
		\IF{$n_{\mathrm{step}}>t_{\mathrm{warm}}$ and $|\Dset|\ge M$}
		\STATE Sample a mini-batch from $\Dset$ and minimize the mini-batch average of \eqref{eq:loss}.
		\STATE Set $n_{\mathrm{upd}}\leftarrow n_{\mathrm{upd}}+1$.
		\IF{$n_{\mathrm{upd}}$ is a multiple of $U$}
		\STATE Set $\theta^-\leftarrow\theta$.
		\ENDIF
		\ENDIF
		\ENDFOR
		\ENDFOR
	\end{algorithmic}
\end{algorithm}

\subsection{Action-Output Dimensionality}

For $n_B$ decision dimensions with $K$ candidate actions per dimension, the Cartesian branch-action set contains
\begin{equation}
	N_{\mathrm{tuple}}
	=
	K^{n_B}
	\label{eq:tuple_complexity}
\end{equation}
pre-projection tuples.
Explicit enumeration of the same symmetric tuple set requires $K^{n_B}$ action-dependent value outputs, whereas the branching representation requires
\begin{equation}
	O_{\mathrm{branch}}
	=
	n_BK
	\label{eq:branch_output_count}
\end{equation}
action-dependent outputs in addition to the shared scalar value stream.

For \method{}, $n_B=4$ and $K=5$, giving $625$ pre-projection tuples and $20$ branch-action outputs.
The eight-branch configuration gives $5^8=390{,}625$ tuples and $40$ branch-action outputs.
These counts describe the symmetric five-choice branch grid before quota projection; multiple tuples may map to the same integer quota allocation.
The Flat DDQN and Flat D3QN comparators use the separate $108$-action catalog specified in Section~\ref{sec:evaluation}.

\section{Simulation Evaluation}\label{sec:evaluation}

\subsection{Evaluation Configuration and Evidence Sets}

The evaluation follows the CBRA model and performance measures defined in Section~\ref{sec:system_problem}.
Each decision interval introduces a fixed number of new UEs, after which the active population evolves through ACB admission, preamble contention, procedure-level resolution, fallback, blocking, and backlog carryover.
The evaluated arrival-load set is $\Lambda=\{10,20,50,100,200\}$, where each value denotes the number of new UEs introduced per decision interval.
Because unresolved UEs remain active through backlog carryover, the active and attempted populations can exceed the nominal arrival load.

Performance measures are recomputed from their accumulated numerators and
denominators according to \eqref{eq:ratio_of_sums}.
For each run, cross-load, seed-sensitivity, and ablation summaries pool
all $100$ decision intervals from epochs $401$--$500$, so
$\Tset_{\mathrm{tail}}=\{401,\ldots,500\}\times\{1,\ldots,100\}$.
The simulator starts each epoch with an empty backlog and static quotas,
while the trainable controllers retain their learned parameters and replay
memory.

The evaluation comprises three complementary analyses.
Cross-load comparison covers the five controller configurations over $\Lambda$ under nominal seed $42$.
Seed sensitivity evaluates \method{} over the same load set using $\mathcal{Z}=\{42,23,16,15,8,4\}$.
The ablation study examines branch-action, reward, warm-up, target-update, and hidden-width variants at loads $20$, $100$, and $200$ under nominal seed $42$.

Table~\ref{tab:sim_params} summarizes the simulation configuration and the
default learning parameters of the trainable controllers.
\begin{table}[!t]\rmfamily\scriptsize
	\centering
	\caption{Simulation and Default Training Configuration}
	\label{tab:sim_params}
	\begin{tabularx}{\columnwidth}
		{@{}>{\raggedright\arraybackslash}p{0.47\columnwidth}Y@{}}
		\toprule
		Parameter & Value \\
		\midrule
		CBRA preambles & $54$ \\
		Arrival loads & $10,20,50,100,200$ \\
		Nominal comparison seed & $42$ \\
		Seed-sensitivity set & $42,23,16,15,8,4$ \\
		M2M probability & $0.8$ \\
		DS probability & $0.3$ for H2H; $0.1$ for M2M \\
		H2H data-amount variable & Discrete uniform on $\{2,\ldots,8\}$ \\
		M2M data-amount variable & Discrete uniform on $\{1,2,3\}$ \\
		Initial access mode & $2\mathrm{RA}$ if data amount $\le 2$; otherwise $4\mathrm{RA}$ \\
		Four-pool static quotas (DS2/DT2/DS4/DT4) & $14/13/14/13$ \\
		$2\mathrm{RA}$/$4\mathrm{RA}$ failure thresholds & $F_2=6$, $F_4=11$ \\
		ACB allow probability & $0.5$ \\
		Default reward profile & Coefficients in \eqref{eq:mild_penalty} \\
		Training epochs & $500$ \\
		Decision intervals per epoch & $100$ \\
		Analysis window & Final $100$ epochs \\
		Warm-up interactions & $2000$ \\
		Learning rate & $3\times10^{-5}$ \\
		Post-warm-up exploration probability & $0.1$ \\
		Mini-batch size & $64$ \\
		Discount factor & $0.90$ \\
		Replay capacity & $5000$ \\
		Target-update period & $500$ gradient updates \\
		Hidden width & $128$ \\
		\bottomrule
	\end{tabularx}
\end{table}
For the trainable controllers, the first $2000$ interactions use uniform
action sampling; subsequent decisions use the fixed exploration probability
listed in Table~\ref{tab:sim_params}.
The same traffic-generation, ACB, access-mode, and contention-resolution model is used throughout the evaluation.

\subsection{Comparator Configuration}

Table~\ref{tab:baseline_surface} defines the five controller configurations used for cross-load comparison and structural analysis.
\begin{table*}[!t]\rmfamily\scriptsize
	\centering
	\caption{Controller Configurations Used in the Evaluation}
	\label{tab:baseline_surface}
	\begin{tabularx}{\textwidth}
		{@{}>{\raggedright\arraybackslash}p{0.16\textwidth}
			>{\raggedright\arraybackslash}p{0.16\textwidth}
			>{\raggedright\arraybackslash}p{0.2\textwidth}
			Y@{}}
		\toprule
		Method & Role & Value representation & Action representation \\
		\midrule
		Static ACB
		& Non-learning comparator
		& Fixed control
		& Fixed ACB and static preamble allocation \\

		Flat DDQN
		& Learning baseline
		& Double DQN
		& Single head over $108$ joint actions \\

		Flat D3QN
		& Learning baseline
		& Dueling Double DQN
		& Single head over the same $108$ joint actions \\

		\method{}
		& Proposed controller
		& Branching Dueling Double DQN
		& Four branches with five actions per branch \\

		B-D3QN-RACH-8
		& Structural comparator
		& Branching Dueling Double DQN
		& Eight branches with five actions each over eight contention groups \\
		\bottomrule
	\end{tabularx}
\end{table*}
Static ACB uses the fixed ACB probability and static preamble allocation.
Flat DDQN and Flat D3QN use the same fixed ACB gate and the same $108$-element joint action catalog over the four preamble pools.
Flat D3QN adds the dueling value--advantage decomposition to the Double DQN architecture.
The catalog is the Cartesian product of the multiplier sets
$\{1,1.2,1.4,1.6\}$ for DS2, $\{0.8,1,1.2\}$ for DT2 and DT4,
and $\{1,1.2,1.4\}$ for DS4, giving $4\times3^3=108$ actions.

\method{} selects one of five multipliers on each of four branches and applies
quota projection before execution, yielding $625$ pre-projection branch-action tuples.
B-D3QN-RACH-8 assigns quotas to eight QoS/access-mode/user-type
contention groups and performs preamble selection separately within each
group. Four-pool quota totals are also recorded for reporting. The
four-pool learning controllers use $28$ observation features; the
eight-group controller uses $35$.
Together, the configurations compare static and learned four-pool
allocation with finer eight-group contention partitioning.

\subsection{Training Behavior and Cross-Load Performance}

Fig.~\ref{fig:training_diagnostics} compares training reward in the upper
row and validation composite score in the lower row across the five arrival
loads under nominal seed $42$. The five plotted epoch positions are
$100$, $200$, $300$, $400$, and $500$.
The validation composite applies the weighted access-quality terms of the training objective to the validation trajectory and provides a consistent view of checkpoint evolution across training.

\begin{figure*}[!t]
	\centering
	\includegraphics[width=0.85\textwidth]{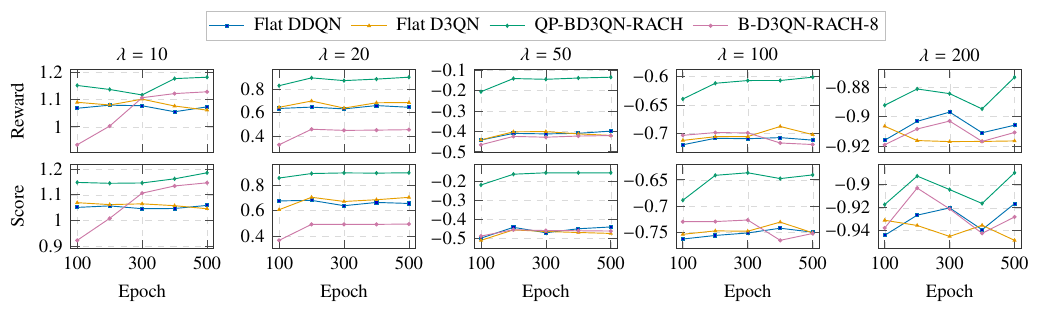}
	\caption{Training reward (upper row) and validation composite score
		(lower row) across the five arrival loads under nominal seed $42$.
		Markers occupy the five plotted epoch positions $100$--$500$.}
	\label{fig:training_diagnostics}
\end{figure*}

Let $m$ denote \method{}. For comparator $b$, metric $k$, and load
$\lambda$, let $R_{k,\lambda}^{m}$ denote the corresponding
tail-window value.
Define $\sigma_k=1$ for higher-is-better metrics and $\sigma_k=-1$ for lower-is-better metrics. The direction-aligned difference is
\begin{equation}
	\Delta_{k,\lambda}^{m-b}
	=
	\sigma_k
	\left(
	R_{k,\lambda}^{m}
	-
	R_{k,\lambda}^{b}
	\right),
	\label{eq:directional_diff}
\end{equation}
and its average over the evaluated load grid is
\begin{equation}
	\bar{\Delta}_{k}^{m-b}
	=
	\frac{1}{|\Lambda|}
	\sum_{\lambda\in\Lambda}
	\Delta_{k,\lambda}^{m-b}.
	\label{eq:mean_directional_diff}
\end{equation}
Positive values therefore indicate better performance by \method{} after accounting for the direction of the metric.

Fig.~\ref{fig:cross_load} shows the five performance measures across the full load range.
Success per attempt and collision per attempt follow the complementarity relation in \eqref{eq:post_cr_identity} and consequently exhibit mirrored trends.
From loads $10$ through $100$, \method{} records the highest success per attempt and the corresponding lowest collision per attempt.
It also records the lowest two-step fallback across all five loads and the lowest successful-access delay at loads $10$, $20$, $100$, and $200$.

\begin{figure*}[t]
	\centering
	\includegraphics[width=0.85\textwidth]{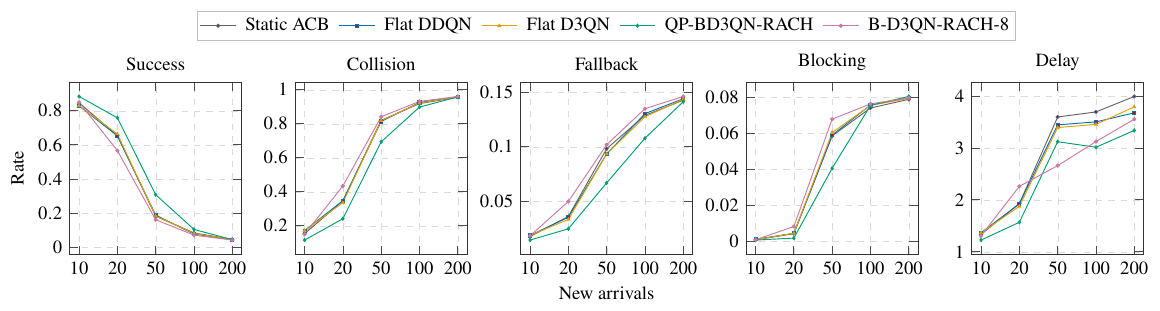}
	\caption{Cross-load performance of the five controller configurations under nominal seed $42$.}
	\label{fig:cross_load}
\end{figure*}

The comparison also reveals two distinct operating points.
At load $50$, B-D3QN-RACH-8 records the lowest successful-access delay, while \method{} records the highest success and the lowest collision, fallback, and blocking.
At load $200$, Static ACB records the highest success and the lowest collision and four-step blocking, while \method{} records the lowest two-step fallback and successful-access delay.
The cross-load curves therefore show how the relative balance among successful resolution, fallback, blocking, and delay changes as contention increases.

Equations~\eqref{eq:directional_diff} and \eqref{eq:mean_directional_diff} transform the five load-wise comparisons in Fig.~\ref{fig:cross_load} to a common better-positive scale and average them over $\Lambda$.
Table~\ref{tab:directional_mean} reports
$100\bar{\Delta}_{k}^{m-b}$ in percentage points for rate measures and
$\bar{\Delta}_{k}^{m-b}$ in decision intervals for delay.
Success and collision have identical direction-aligned differences because of \eqref{eq:post_cr_identity}.

\begin{table}[!t]\rmfamily\scriptsize
	\centering
	\caption{Mean Direction-Aligned Differences Across Five Loads:
		Rates in Percentage Points and Delay in Decision Intervals}
	\label{tab:directional_mean}
	\begin{tabular*}{\columnwidth}
		{@{\extracolsep{\fill}}lrrrr@{}}
		\toprule
		Comparator & Success/Collision & Fallback & Blocking & Delay \\
		\midrule
		Static ACB      & 5.74 & 1.37 & 0.35 & 0.456 \\
		Flat DDQN       & 6.18 & 1.34 & 0.42 & 0.329 \\
		Flat D3QN       & 5.95 & 1.23 & 0.43 & 0.323 \\
		B-D3QN-RACH-8   & 8.21 & 1.92 & 0.68 & 0.128 \\
		\bottomrule
	\end{tabular*}
\end{table}

Table~\ref{tab:directional_mean} shows positive mean differences for all four comparators across the evaluated load grid, with the largest success/collision and fallback differences occurring against B-D3QN-RACH-8.
Fig.~\ref{fig:directional_stats} reports the means with descriptive
$95\%$ percentile-bootstrap intervals obtained from $2000$ resamples,
with replacement, of the five load-wise paired differences. Its rate
axes show unscaled differences, with $0.01$ corresponding to one
percentage point.

\begin{figure*}[t]
	\centering
	\includegraphics[width=0.85\textwidth]{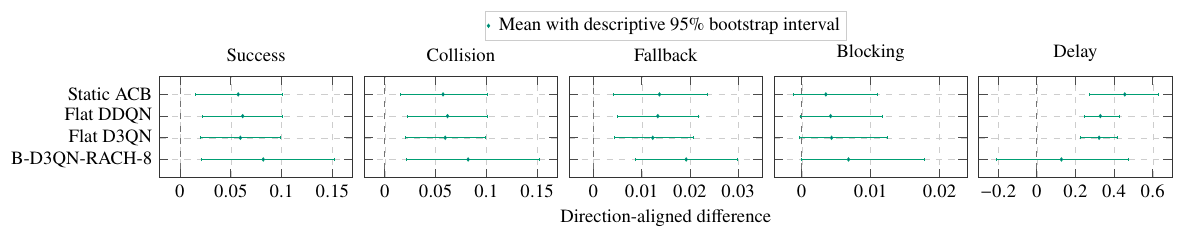}
	\caption{Direction-aligned differences from each comparator under
		nominal seed $42$. Markers show five-load means; horizontal bars show
		descriptive $95\%$ percentile-bootstrap intervals of the means.
		Positive values favor \method{}. Rate differences are unscaled;
		delay differences are measured in decision intervals.}
	\label{fig:directional_stats}
\end{figure*}

\subsection{Seed Sensitivity and Ablation Analysis}

Seed sensitivity of \method{} is evaluated over $\mathcal{Z}=\{42,23,16,15,8,4\}$.
Let $R_{k,\lambda,z}$ denote the tail-window value of \method{} for
metric $k$, load $\lambda$, and seed $z$. The six-seed mean is
\begin{equation}
	\bar{R}_{k,\lambda}
	=
	\frac{1}{|\mathcal{Z}|}
	\sum_{z\in\mathcal{Z}}R_{k,\lambda,z}.
	\label{eq:seed_mean}
\end{equation}

Fig.~\ref{fig:seed_sensitivity} reports the six seed curves and their
mean. At each load, vertical error bars show $95\%$ percentile-bootstrap
intervals of the mean, obtained from $2000$ resamples with replacement
of the six seed-level values.
Seed variation is comparatively small at loads $10$, $20$, $100$, and $200$, whereas load
$50$ exhibits the widest spread, particularly in success/collision and delay.

\begin{figure*}[t]
	\centering
	\includegraphics[width=0.9\textwidth]{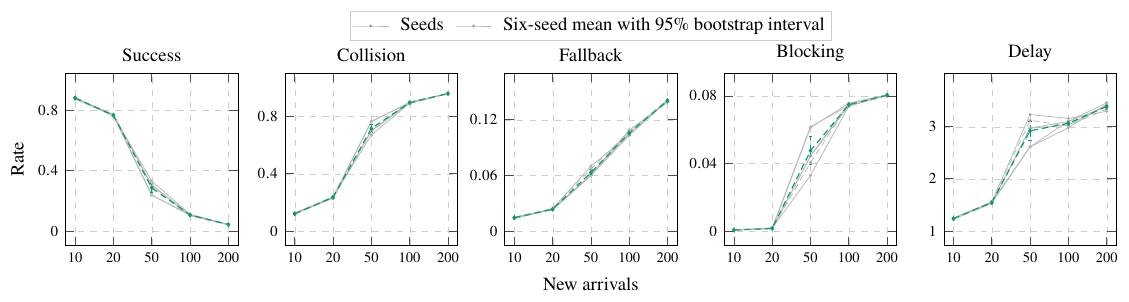}
	\caption{Seed sensitivity of \method{} across six seeds and five
		arrival loads. Gray curves show individual seeds; green diamonds show
		their mean, with vertical bars denoting $95\%$ percentile-bootstrap
		intervals of the mean.}
	\label{fig:seed_sensitivity}
\end{figure*}

The ablation analysis evaluates the ``Branch mild'' grid
$\{0.75,0.875,1,1.125,1.25\}$ and the ``Branch narrow'' grid
$\{0.875,0.9375,1,1.0625,1.125\}$, together with the full-penalty reward,
warm-up $1000$, target-update $200$, and hidden-width $256$ variants at
loads $20$, $100$, and $200$. The full-penalty profile retains the success
coefficients and doubles every penalty coefficient in
\eqref{eq:mild_penalty}.
For variant $v$, define
\begin{equation}
	\Delta_{k,\lambda}^{v-\mathrm{def}}
	=
	\sigma_k
	\left(
	R_{k,\lambda}^{v}
	-
	R_{k,\lambda}^{\mathrm{def}}
	\right).
	\label{eq:ablation_contrast}
\end{equation}
Values above zero favor the variant under the metric-direction convention, whereas values below zero favor the default configuration.

As shown in Fig.~\ref{fig:ablation}, narrowing the branch-action range shifts the operating balance toward selected high-load success, collision, and
blocking measures, while the default branch grid gives stronger low- and mid-load success, fallback, and delay behavior.
The full-penalty reward, shorter warm-up, faster target update, and wider hidden layer produce more localized changes across the three loads.
These load- and metric-dependent tradeoffs support retaining the default
five-choice branch grid, mild-penalty reward, $2000$-interaction warm-up,
$500$-update target period, and hidden width $128$ for the main comparison.

\begin{figure*}[!t]
	\centering
	\includegraphics[width=0.85\textwidth]{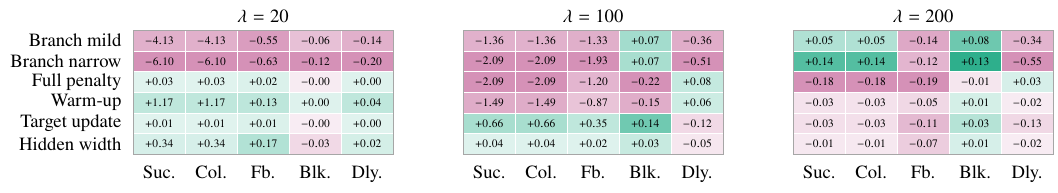}
	\caption{Direction-aligned ablation differences relative to the default \method{} configuration at loads $20$, $100$, and $200$.
		Rate entries are percentage-point differences; delay entries are
		differences in decision intervals. Values are rounded to two decimal places.}
	\label{fig:ablation}
\end{figure*}

\subsection{Mechanism Diagnostics and Representation Structure}

The aggregate performance trends can be related to the underlying access pressure through Fig.~\ref{fig:kpi_backlog}.
The figure compares attempted UEs per preamble, pre-resolution contention
pressure, mean post-interval backlog, and nominal message cost per
successful UE across controller configurations and arrival loads.
Together with Fig.~\ref{fig:cross_load}, these quantities describe how
access pressure, post-interval backlog, and nominal signaling cost vary
across arrival loads.

\begin{figure*}[t]
	\centering
	\includegraphics[width=0.85\textwidth]{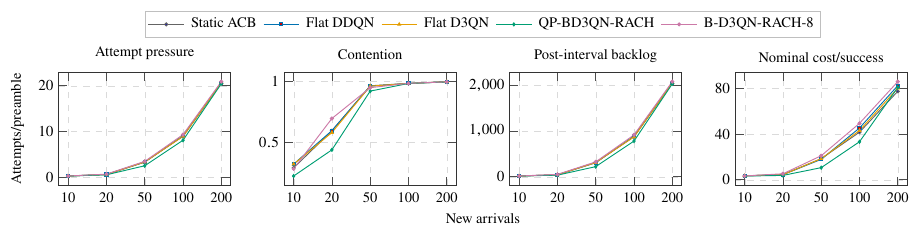}
	\caption{Cross-load access diagnostics: attempted UEs per preamble,
		pre-resolution contention pressure, mean post-interval backlog, and
		nominal message cost per successful UE.}
	\label{fig:kpi_backlog}
\end{figure*}

Fig.~\ref{fig:mode_qos} reports 2RA/4RA success and delay, together with
DT-2RA and M2M-2RA success. Success panels use current access mode and
the corresponding active-user denominators; delay panels group successful
UEs by initial access mode.

\begin{figure*}[t]
	\centering
	\includegraphics[width=\textwidth]{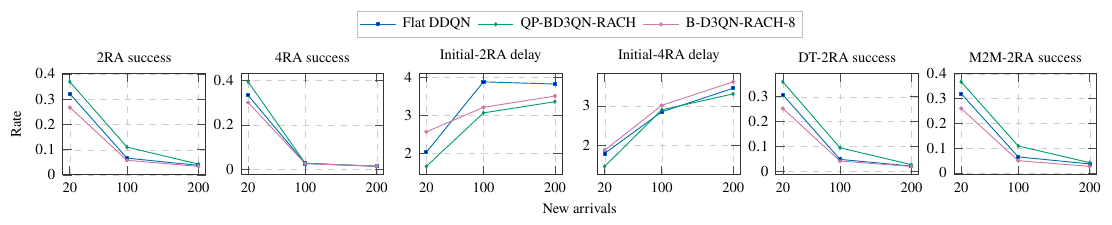}
	\caption{Selected access-mode and traffic measures at loads $20$,
		$100$, and $200$ for Flat DDQN, \method{}, and B-D3QN-RACH-8.
		Success rates use current-mode active populations; delays use
		successful UEs grouped by initial access mode.}
	\label{fig:mode_qos}
\end{figure*}

Fig.~\ref{fig:state_action_reward_correlation} reports tail-window Spearman
associations between aggregate access/composition variables and projected
quota shares, reward, and failure pressure.

\begin{figure}[!t]
	\centering
	\includegraphics[width=0.9\columnwidth]{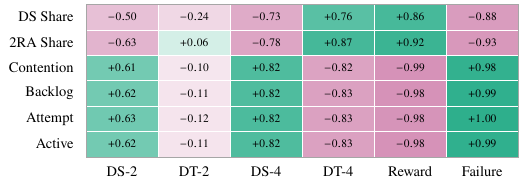}
	\caption{Tail-window Spearman associations between aggregate
		access/composition variables and projected quota shares, reward, and failure
		pressure for \method{}. Values are rounded to two decimal places.}
	\label{fig:state_action_reward_correlation}
\end{figure}

Finally, Table~\ref{tab:complexity} complements the behavioral diagnostics with the value-representation dimensions derived in \eqref{eq:tuple_complexity} and \eqref{eq:branch_output_count}.
For the five-choice grid, \method{} represents $625$ pre-projection branch-action tuples with $20$ branch-action outputs.
The eight-branch configuration uses $40$ outputs for $390{,}625$ tuples, while explicit enumeration of the symmetric four-pool Cartesian reference requires $625$ action-dependent outputs.

\begin{table}[!t]\rmfamily\scriptsize
	\centering
	\caption{Action-Value Output Counts for the Five-Choice Branch Grid}
	\label{tab:complexity}
	\begin{tabular*}{\columnwidth}
		{@{\extracolsep{\fill}}lrr@{}}
		\toprule
		Representation & Explicit outputs & Cartesian tuples \\
		\midrule
		Symmetric flat four-pool & 625 & 625 \\
		\method{} & 20 & 625 \\
		B-D3QN-RACH-8 & 40 & 390,625 \\
		\bottomrule
	\end{tabular*}
\end{table}

\section{Conclusion}\label{sec:conclusion}

This paper presented \method{}, a quota-projected branching reinforcement-learning
controller for finite-budget QoS-aware CBRA preamble slicing. The four action
branches follow the four executable QoS/access-mode quota dimensions, while
deterministic projection converts the selected multipliers to feasible integer
allocations. Across the evaluated load grid, \method{} achieves favorable mean
direction-aligned differences against all four comparators and exhibits
load-dependent performance gains in success, collision, fallback, blocking,
and successful-access delay. At load $200$, Static ACB achieves higher success
and lower collision and four-step blocking, while \method{} maintains lower
two-step fallback and successful-access delay.
The evaluated four-pool configuration combines budget-feasible
allocation with a compact action-value representation and load-dependent
performance tradeoffs.

The current evidence is bounded by the adopted simulation and comparison design.
Cross-method comparisons use nominal seed $42$, and the six-seed
analysis concerns \method{}. The flat and four-branch learning controllers share the same aggregate observation but use different action representations, while B-D3QN-RACH-8 additionally employs a finer user-type-specific observation and contention-group partition.
The CBRA model performs procedure-level
contention resolution and does not include physical-layer capture,
signal-to-interference-plus-noise ratio (SINR),
power control, multi-user decoding, contention-free access, or MsgA/Msg3
decoding. The load-paired intervals and Spearman coefficients provide descriptive
summaries of the evaluated data. Future work should extend the evaluation
to multi-seed cross-method comparisons, broader traffic and arrival processes,
external baseline reproductions, and channel-aware access models.

\bibliographystyle{IEEEtran}
\bibliography{references}

@techreport{3gpp.22.011,
	author = {3GPP},
	institution = {3rd Generation Partnership Project},
	number = {22.011},
	note = {Version 18.6.0},
	year = {2025},
	title = {{Technical Specification Group Services and System Aspects; Service accessibility}},
	type = {Technical Specification (TS)},
	url = {https://www.3gpp.org/dynareport/22011.htm}
}

@article{althumali2022priority,
	title={Priority-based load-adaptive preamble separation random access for {QoS}-differentiated services in {5G} networks},
	author={Althumali, Huda and Othman, Mohamed and Noordin, Nor Kamariah and Hanapi, Zurina Mohd},
	journal={Journal of Network and Computer Applications},
	volume={203},
	pages={103396},
	year={2022},
	publisher={Elsevier},
	doi={10.1016/j.jnca.2022.103396}}

@article{du2020machine,
	title={Machine learning for {6G} wireless networks: Carrying forward enhanced bandwidth, massive access, and ultrareliable/low-latency service},
	author={Du, Jun and Jiang, Chunxiao and Wang, Jian and Ren, Yong and Debbah, Merouane},
	journal={IEEE Vehicular Technology Magazine},
	volume={15},
	number={4},
	pages={122--134},
	year={2020},
	publisher={IEEE},
	doi={10.1109/MVT.2020.3019650}}

@techreport{3gpp.38.300,
	author = {3GPP},
	institution = {3rd Generation Partnership Project},
	number = {38.300},
	note = {Version 18.6.0},
	year = {2025},
	title = {{Technical Specification Group Radio Access Network; NR; NR and NG-RAN Overall Description; Stage 2}},
	type = {Technical Specification (TS)},
	url = {https://www.3gpp.org/dynareport/38300.htm}
}

@techreport{3gpp.38.321,
	author = {3GPP},
	institution = {3rd Generation Partnership Project},
	number = {38.321},
	note = {Version 18.6.0},
	year = {2025},
	title = {{Technical Specification Group Radio Access Network; NR; Medium Access Control (MAC) protocol specification}},
	type = {Technical Specification (TS)},
	url = {https://www.3gpp.org/dynareport/38321.htm}
}

@article{chowdhury2022queue,
	title={Queue-aware access prioritization for massive machine-type communication},
	author={Chowdhury, Mayukh Roy and De, Swades},
	journal={IEEE Internet of Things Journal},
	volume={9},
	number={17},
	pages={15858--15873},
	year={2022},
	publisher={IEEE},
	doi={10.1109/JIOT.2022.3151408}}

@article{lee2021enhanced,
	title={Enhanced random access for massive-machine-type communications},
	author={Lee, Byung-Hyun and Lee, Hyun-Suk and Moon, Seokjae and Lee, Jang-Won},
	journal={IEEE Internet of Things Journal},
	volume={8},
	number={8},
	pages={7046--7064},
	year={2021},
	publisher={IEEE},
	doi={10.1109/JIOT.2020.3038148}}

@article{pacheco2019deep,
	title={Deep reinforcement learning mechanism for dynamic access control in wireless networks handling {mMTC}},
	author={Pacheco-Paramo, Diego and Tello-Oquendo, Luis and Pla, Vicent and Martinez-Bauset, Jorge},
	journal={Ad Hoc Networks},
	volume={94},
	pages={101939},
	year={2019},
	publisher={Elsevier},
	doi={10.1016/j.adhoc.2019.101939}}

@article{jiang2019multiple,
	title={Multiple preambles for high success rate of grant-free random access with massive {MIMO}},
	author={Jiang, Hao and Qu, Daiming and Ding, Jie and Jiang, Tao},
	journal={IEEE Transactions on Wireless Communications},
	volume={18},
	number={10},
	pages={4779--4789},
	year={2019},
	publisher={IEEE},
	doi={10.1109/TWC.2019.2929126}}

@article{he2022cluster,
	title={Cluster-aided collision resolution random access in distributed massive {MIMO} systems},
	author={He, Yuxuan and Ren, Guangliang},
	journal={IEEE Internet of Things Journal},
	volume={9},
	number={13},
	pages={11453--11463},
	year={2022},
	publisher={IEEE},
	doi={10.1109/JIOT.2021.3127936}}

@article{piao2024integrated_ra,
	author={Piao, Yijun and Lee, Tae Jin},
	title={Integrated 2--4 Step Random Access for Heterogeneous and Massive {IoT} Devices},
	journal={IEEE Transactions on Green Communications and Networking},
	volume={8},
	number={1},
	pages={441--452},
	year={2024},
	doi={10.1109/TGCN.2023.3322539},
	publisher={IEEE}}

@article{khairy2022data,
	title={Data-driven random access optimization in multi-cell {IoT} networks using {NOMA}},
	author={Khairy, Sami and Balaprakash, Prasanna and Cai, Lin X and Poor, H Vincent},
	journal={IEEE Transactions on Wireless Communications},
	volume={21},
	number={7},
	pages={4938--4953},
	year={2022},
	publisher={IEEE},
	doi={10.1109/TWC.2021.3134949}}

@article{liva2024unsourced,
	title={Unsourced multiple access: A coding paradigm for massive random access},
	author={Liva, Gianluigi and Polyanskiy, Yury},
	journal={Proceedings of the IEEE},
	volume={112},
	number={9},
	pages={1214--1229},
	year={2024},
	publisher={IEEE},
	doi={10.1109/JPROC.2024.3437208}}

@article{mnih2015human,
	title={Human-level control through deep reinforcement learning},
	author={Mnih, Volodymyr and Kavukcuoglu, Koray and Silver, David and Rusu, Andrei A and Veness, Joel and Bellemare, Marc G and Graves, Alex and Riedmiller, Martin and Fidjeland, Andreas K and Ostrovski, Georg and others},
	journal={Nature},
	volume={518},
	number={7540},
	pages={529--533},
	year={2015},
	publisher={Nature Publishing Group},
	doi={10.1038/nature14236}
}

@article{liu2024enhanced,
  title={Enhanced {RACH} optimization in {IoT} networks: A {DQN} approach for balancing {H2H} and {M2M} communications},
  author={Liu, Xue and Yang, Heng and Li, Shanshan and Liu, Zhenyu and Lian, Xiaohui},
  journal={Internet of Things},
  volume={28},
  pages={101433},
  year={2024},
  publisher={Elsevier},
  doi={10.1016/j.iot.2024.101433}}

@article{pacheco2020delay,
  title={Delay-aware dynamic access control for {mMTC} in wireless networks using deep reinforcement learning},
  author={Pacheco-Paramo, Diego and Tello-Oquendo, Luis},
  journal={Computer Networks},
  volume={182},
  pages={107493},
  year={2020},
  publisher={Elsevier},
  doi={10.1016/j.comnet.2020.107493}}

@article{fan2024joint_delay_energy,
  author={Fan, Wenbo and Fan, Pingzhi and Long, Yan},
  title={Joint Delay-Energy Optimization for Multi-Priority Random Access in Machine-Type Communications},
  journal={IEEE Transactions on Wireless Communications},
  volume={23},
  number={2},
  pages={1416--1431},
  year={2024},
  doi={10.1109/TWC.2023.3289314},
  publisher={IEEE}}

@article{sohaib2022dynamic,
  title={Dynamic Multichannel Access via Multi-Agent Reinforcement Learning: Throughput and Fairness Guarantees},
  author={Sohaib, Muhammad and Jeong, Jongjin and Jeon, Sang-Woon},
  journal={IEEE Transactions on Wireless Communications},
  volume={21},
  number={6},
  pages={3994--4008},
  year={2022},
  publisher={IEEE},
  doi={10.1109/TWC.2021.3126112}}

@article{nie2025noma_rach,
  title={A {NOMA}-Enhanced Two-Step {RACH} Procedure for Low-Latency Access in {5G} Networks},
  author={Nie, Dawei and Yu, Wenjuan and Foh, Chuan Heng and Ni, Qiang},
  journal={IEEE Internet of Things Journal},
  volume={12},
  number={9},
  pages={11568--11580},
  year={2025},
  publisher={IEEE},
  doi={10.1109/JIOT.2024.3521340}}

@article{marinello2025grantfree_ris,
  title={Grant-Free Random Access for {RIS}-Aided Machine-Type Communication},
  author={{Marinello Filho}, Jos{\'e} Carlos and Abr{\~a}o, Taufik and Hossain, Ekram and Mezghani, Amine},
  journal={IEEE Transactions on Wireless Communications},
  volume={24},
  number={9},
  pages={7794--7808},
  year={2025},
  publisher={IEEE},
  doi={10.1109/TWC.2025.3562959}}

@article{zhang2026delay_optimal_ra,
  title={Delay-Optimal Random Access: A Learning Framework},
  author={Zhang, Huaqiang and Zhao, Xinran and Dai, Lin},
  journal={IEEE Transactions on Communications},
  volume={74},
  pages={3059--3073},
  year={2026},
  publisher={IEEE},
  doi={10.1109/TCOMM.2026.3652495}}

@inproceedings{tavakoli2018action_branching,
  title={Action Branching Architectures for Deep Reinforcement Learning},
  author={Tavakoli, Arash and Pardo, Fabio and Kormushev, Petar},
  booktitle={Proceedings of the AAAI Conference on Artificial Intelligence},
  volume={32},
  number={1},
  pages={4131--4138},
  year={2018},
  doi={10.1609/aaai.v32i1.11798}}

@article{jiang2021decoupled_ra,
  author={Jiang, Nan and Deng, Yansha and Nallanathan, Arumugam and Yuan, Jinhong},
  title={A Decoupled Learning Strategy for Massive Access Optimization in Cellular {IoT} Networks},
  journal={IEEE Journal on Selected Areas in Communications},
  volume={39},
  number={3},
  pages={668--685},
  year={2021},
  doi={10.1109/JSAC.2020.3018806},
  publisher={IEEE}}

@inproceedings{tello2018rl_acb,
  title={Reinforcement Learning-Based {ACB} in {LTE-A} Networks for Handling Massive {M2M} and {H2H} Communications},
  author={Tello-Oquendo, Luis and Pacheco-Paramo, Diego and Pla, Vicent and Martinez-Bauset, Jorge},
  booktitle={2018 IEEE International Conference on Communications},
  pages={1--7},
  year={2018},
  organization={IEEE},
  doi={10.1109/ICC.2018.8422167}}

@article{bui2020energy_acb,
  title={Deep Reinforcement Learning-Based Access Class Barring for Energy-Efficient {mMTC} Random Access in {LTE} Networks},
  author={Bui, Anh-Tuan H. and Pham, Anh T.},
  journal={IEEE Access},
  volume={8},
  pages={227657--227666},
  year={2020},
  publisher={IEEE},
  doi={10.1109/ACCESS.2020.3045811}}

@article{bai2021multiagent_ra,
  author={Bai, Jianan and Song, Hao and Yi, Yang and Liu, Lingjia},
  title={Multiagent Reinforcement Learning Meets Random Access in Massive Cellular Internet of Things},
  journal={IEEE Internet of Things Journal},
  volume={8},
  number={24},
  pages={17417--17428},
  year={2021},
  doi={10.1109/JIOT.2021.3081692},
  publisher={IEEE}}

@inproceedings{turan2021adaptive_acb_slicing,
  title={Reinforcement Learning Based Adaptive Access Class Barring for {RAN} Slicing},
  author={Turan, Ali and Koseoglu, Mehmet and Sezer, Ebru Akcapinar},
  booktitle={2021 IEEE International Conference on Communications Workshops},
  pages={1--6},
  year={2021},
  organization={IEEE},
  doi={10.1109/ICCWorkshops50388.2021.9473684}}

@article{jang2021cellular_ra_framework,
  title={Deep Learning-Based Cellular Random Access Framework},
  author={Jang, Han Seung and Lee, Hoon and Quek, Tony Q. S. and Shin, Hyundong},
  journal={IEEE Transactions on Wireless Communications},
  volume={20},
  number={11},
  pages={7503--7518},
  year={2021},
  publisher={IEEE},
  doi={10.1109/TWC.2021.3085303}}

@article{gedikli2022flexible_preamble,
  title={Deep reinforcement learning based flexible preamble allocation for {RAN} slicing in {5G} networks},
  author={Gedikli, Ahmet Melih and Koseoglu, Mehmet and Sen, Sevil},
  journal={Computer Networks},
  volume={215},
  pages={109202},
  year={2022},
  publisher={Elsevier},
  doi={10.1016/j.comnet.2022.109202}}

@article{elmeligy2025preamble_mab,
  title={Preamble Selection Probability Optimization in {RACH}: A Multi-Armed Bandits Approach},
  author={Elmeligy, Ahmed O. and Psaromiligkos, Ioannis and Au, Minh},
  journal={IEEE Open Journal of the Communications Society},
  volume={6},
  pages={10761--10780},
  year={2025},
  publisher={IEEE},
  doi={10.1109/OJCOMS.2025.3647566}}

@inproceedings{vanhasselt2016deep,
  title={Deep Reinforcement Learning with Double {Q}-Learning},
  author={Van Hasselt, Hado and Guez, Arthur and Silver, David},
  booktitle={Proceedings of the AAAI Conference on Artificial Intelligence},
  volume={30},
  number={1},
  pages={2094--2100},
  year={2016},
  doi={10.1609/aaai.v30i1.10295},
  url={https://ojs.aaai.org/index.php/AAAI/article/view/10295}}

@inproceedings{wang2016dueling,
  title={Dueling Network Architectures for Deep Reinforcement Learning},
  author={Wang, Ziyu and Schaul, Tom and Hessel, Matteo and Van Hasselt, Hado and Lanctot, Marc and De Freitas, Nando},
  booktitle={Proceedings of the 33rd International Conference on Machine Learning},
  pages={1995--2003},
  year={2016},
  volume={48},
  series={Proceedings of Machine Learning Research},
  publisher={PMLR},
  url={https://proceedings.mlr.press/v48/wangf16.html}}

@ARTICLE{wang2018deep,
	author={Wang, Shangxing and Liu, Hanpeng and Gomes, Pedro Henrique and Krishnamachari, Bhaskar},
	journal={IEEE Transactions on Cognitive Communications and Networking},
	title={Deep Reinforcement Learning for Dynamic Multichannel Access in Wireless Networks},
	year={2018},
	volume={4},
	number={2},
	pages={257-265},
	doi={10.1109/TCCN.2018.2809722}}

@ARTICLE{nisioti2019fast,
	author={Nisioti, Eleni and Thomos, Nikolaos},
	journal={IEEE Transactions on Cognitive Communications and Networking},
	title={Fast Q-Learning for Improved Finite Length Performance of Irregular Repetition Slotted ALOHA},
	year={2020},
	volume={6},
	number={2},
	pages={844-857},
	doi={10.1109/TCCN.2019.2957224}}

\end{document}